%% file: cocoon2016.tex
\documentclass{llncs}
\usepackage[pdftex]{graphicx}
\usepackage{wrapfig}
\usepackage{subcaption}
\usepackage{amsmath}
\usepackage{algorithmic}
\PassOptionsToPackage{hyphens}{url}
\usepackage{url}

\begin{document}

\newcommand{\junk}[1]{}

\title{Balls and Funnels: Energy Efficient Group-to-Group Anycasts}
\titlerunning{Balls and Funnels}  
%
\author{Jennifer Iglesias\inst{1}  Rajmohan Rajaraman\inst{2} \and R. Ravi\inst{1} \and Ravi Sundaram\inst{2}} 
\authorrunning{Iglesias, Rajaraman, Ravi, Sundaram} 
%
%
\institute{Carnegie Mellon University, Pittsburgh PA USA \\
\email{\{jiglesia, ravi\}@andrew.cmu.edu }
\and
Northeastern University, Boston MA USA\\
\email{\{rraj, koods\}@ccs.neu.edu}}

\maketitle              

\begin{abstract}

We introduce  group-to-group anycast  (g2g-anycast), a  network design
problem   of  substantial   practical   importance  and   considerable
generality.   Given  a  collection  of  groups  and  requirements  for
directed connectivity  from source  groups to destination  groups, the
solution   network   must   contain,    for   each   requirement,   an
omni-directional  down-link broadcast,  centered  at any  node of  the
source group,  called the ball; the  ball must contain some  node from
the  destination group  in the  requirement and  all such  destination
nodes in  the ball  must aggregate  into a  tree directed  towards the
source, called the funnel-tree.  The  solution network is a collection
of balls along with the funnel-trees they contain.  g2g-anycast models
DBS (Digital Broadcast Satellite), Cable  TV systems and drone swarms.
It generalizes  several well  known network design  problems including
minimum   energy   unicast,    multicast,   broadcast,   Steiner-tree,
Steiner-forest  and Group-Steiner  tree.  Our  main achievement  is an
$O(\log^4     n)$     approximation,     counterbalanced     by     an
$\log^{(2-\epsilon)}n$ hardness of approximation, for general weights.
Given  the  applicability  to  wireless communication,  we  present  a
scalable and  easily implemented $O(\log n)$  approximation algorithm,
Cover-and-Grow  for fixed-dimensional  Euclidean space  with path-loss
exponent at least 2.

\end{abstract}
\input{intro-g2g.tex}
\input{general-g2g.tex}
\input{singleton-g2g.tex}
\input{l2squared-g2g.tex}
\input{shortsimulations.tex}
%
%
\bibliographystyle{splncs03}
\bibliography{group-refs}

\appendix
\input{setCoverreduction.tex}
\input{simulations-g2g.tex}

\end{document}

%% file: intro-g2g.tex
\section{Introduction}
\subsection{Motivation}

Consider a  DBS (Digital Broadcast  Satellite) system such as  Dish or
DIRECTV  in the  USA (see  Fig.~\ref{fig:DBS}).  The  down-link is  an
omni-directional broadcast from constellations of satellites to groups
of  apartments  or   neighborhoods  serviced  by  one   or  more  dish
installations.  The up-link is sometimes a wired network but in remote
areas it is usually structured  as a tree consisting of point-to-point
wireless  links  directed  towards  the  network  provider's  head-end
(root).   The  high  availability  requirement of  such  services  are
typically  satisfied by  having multiple  head-ends and  anycasting to
them. The  same architecture  is found  in CATV  (originally Community
Antenna  TV), or  cable  TV  distribution systems  as  well as  sensor
networks where an omni-directional broadcast  from a beacon is used to
activate  and  control the  sensors;  the  sensors then  funnel  their
information back  using relays.   Moreover, this architecture  is also
beginning to emerge in drone  networks, for broadcasting the Internet,
by   companies  such   as   Google   \cite{McNeal14}  and   Facebook's
Connectivity Labs  \cite{Lapowsky14}. The Internet is  to be broadcast
from drones flying fixed patterns in  the sky to a collection of homes
on the ground. The Internet up-link  from the homes is then aggregated
using  wireless links  organized as  a  tree to  be sent  back to  the
drones. Anycasting  is an integral part  of high-availability services
such as  Content Delivery Networks (CDNs)  where reliable connectivity
is achieved  by reaching some node  in the group.  What  is the common
architecture  underlying  all  these  applications  and  what  is  the
constraining  resource  that  is  driving their  form?\\  \indent  The
various distribution  systems can be  abstractly seen to consist  of a
down-link   \emph{ball}  and   an   up-link  \emph{funnel-tree}   (see
Fig.~\ref{fig:DBS}).  The ball is an omni-directional
\begin{wrapfigure}{r}{0.5\textwidth}
\vspace{-20pt}
\includegraphics[width=0.45\textwidth]{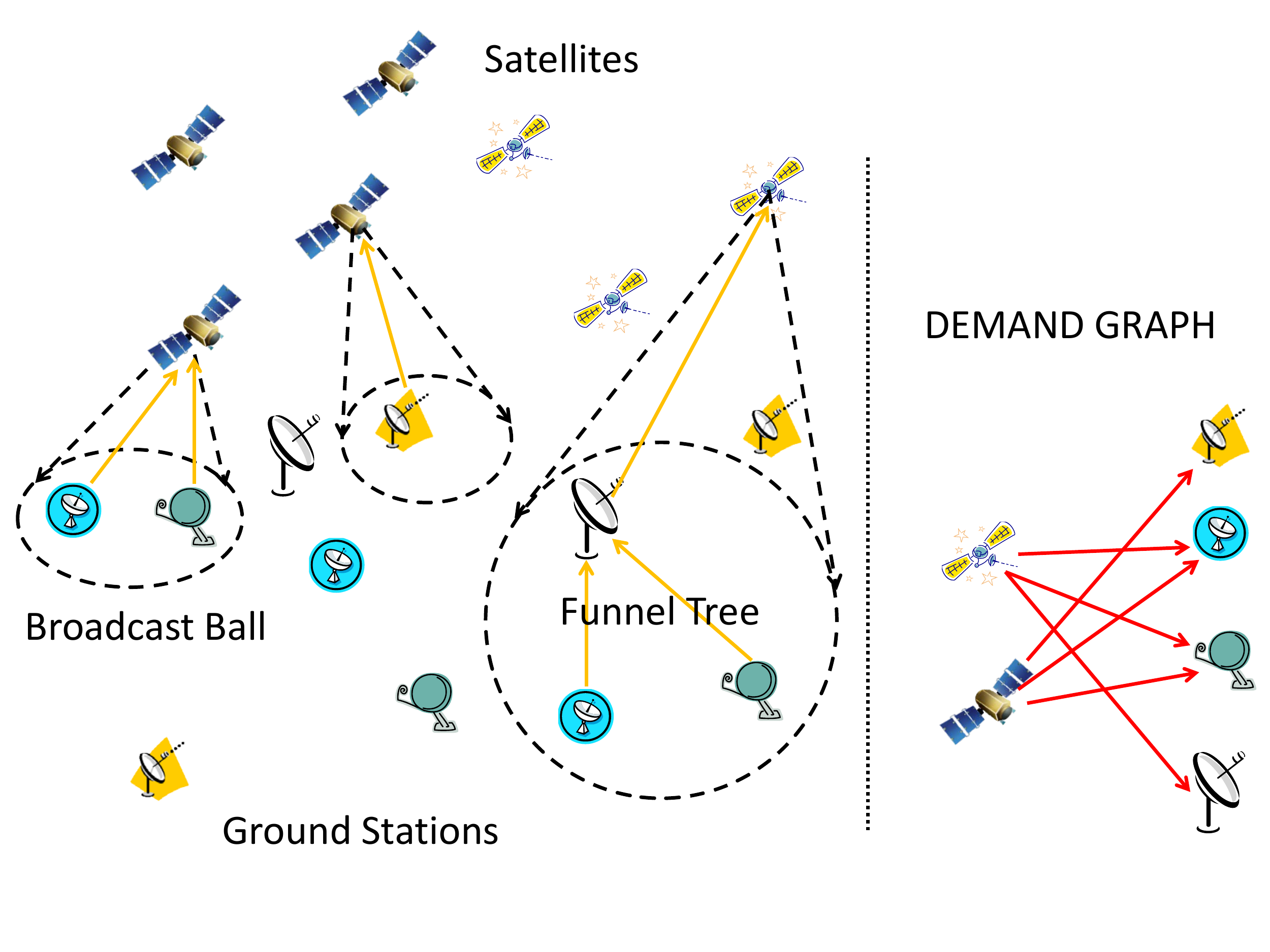}
\caption{Pictogram  of  Digital  Broadcast  Satellite  System  with  2
  satellite groups and 4 ground station groups on left with associated
  demand  graph on  the right.   The  broadcast balls  are denoted  by
  dotted black lines, and the funnel trees by solid yellow lines }
\label{fig:DBS}
\vspace{-15pt}
\end{wrapfigure}
broadcast from
the publisher or content-producer to a large collection of subscribers
or  content-consumers.    At  the   same  time,  the   consumers  have
information that they  need to dynamically send back  to the publisher
in order to convey their preferences and requirements. The funnel-tree
achieves this  up-link efficiently in  terms of both time  and energy.
Aggregation  of information  and use  of  relays uses  less energy  as
compared  to omni-directional  broadcasts  by each  node  back to  the
publisher   and   also  avoids   the   scheduling   needed  to   avoid
interference.  In  this  work,  we focus  primarily  on  total  energy
consumption.   The  application  scenarios mentioned  in  the  opening
paragraph    are    all     energy    sensitive.    Sensor    networks
\cite{MilyekovskiSS13}   and   drone    fleets   \cite{Olsson11}   are
particularly  vulnerable  to energy  depletion.   For  the purpose  of
energy  conservation, generally  each  wireless  node can  dynamically
adjust its transmitting  power based on the distance  of the receiving
nodes and background noise. In the most common power-attenuation model
\cite{Rappaport01},  the signal  power  falls as  $\frac{1}{r^\kappa}$
where $r$  is the distance  from the  transmitter to the  receiver and
$\kappa$  is the  path-loss  exponent -  a constant  between  2 and  4
dependent  on   the  wireless  environment.   A   key  implication  of
non-linear power attenuation is  that relaying through an intermediate
node can sometimes be more energy efficient than transmitting directly
- a counter-intuitive violation of the  triangle inequality - e.g., in
a triangle $ABC$  with obtuse angle $ABC$,  where $d_{AB}^2+d_{BC}^2 <
d_{AC}^2$.

\subsection{Problem Formulation and Terminology}
\label{subsec:formulation}
In this paper, we consider a general formulation that encompasses a
wide variety of scenarios: given a collection of groups (of nodes)
along with a directed demand graph over these groups the goal is to
design a collection of balls and associated funnel-trees of lowest
cost so that every demand requirement is met - meaning that if there
is an arc from a source group to a destination group then the solution
must have a ball centered at a node of the source group that includes
a funnel-tree containing a node of the destination group.

Formally, we define the {\em  group-to-group anycast} problem, or {\em
  g2g-anycast}, as follows: as input we are given $n$ nodes along with
a collection of  source groups $S_1, S_2,\ldots,S_p$  and a collection
of destination groups $T_1,T_2,\ldots,T_q$  which are subsets of these
nodes; a demand graph on these groups consisting of directed arcs from
source groups  $S_i$ to destination  groups $T_j$.  A  nonegative cost
$c_{uv}$ is  specified between every  pair of  nodes; when a  node $u$
incurs a  cost $C$ in  doing an omni-directional broadcast  it reaches
all nodes  $v$ such that $c_{uv}  \leq C$.  A metric  $d_{uv}$ is also
specified between every pair of nodes  and when a node $u$ connects to
node $v$  in the funnel-tree using  a point-to-point link it  incurs a
cost $d_{uv}$.  A  solution consists of a broadcast  ball around every
source node $s$ (we give a  radius which the source can broadcast to),
and a  funnel tree rooted at  $s$. A demand $S_i,T_j$  is satisfied if
there is  a broadcast ball  from some  $s\in S_i$ which  contains some
$t\in T_j$ and the  funnel tree of $s$ also includes  $t$. The cost of
the solution  is the  sum of  the ball-radii  around the  source nodes
(under the broadcast costs $c$) and the sum of the costs of the funnel
trees (under  the funnel metric  $d$) that connect  all terminal-nodes
used to cover the demands to  the source nodes within whose balls they
lie.  We do  not allow funnel trees  to share edges (even  if they are
going to the same source group), and will pay for each copy of an edge
used.
\begin{itemize}
\item First,  the bipartite demand  graph is  no less general  than an
  arbitrary demand  graph since  a given  group can  be both  a source
  group and destination group.
\item Second, since funnel trees sharing the same edge pay seperately,
  solutions  to the  problem decompose  across the  sources and  it is
  sufficient to solve the case where  we have exactly one source group
  $S   =   \{s_1,   s_2,\ldots,    s_k\}$   and   destination   groups
  $T_1,T_2,\ldots,T_q$ (i.e. the demand graph  is a star consisting of
  all  arcs $(S,  T_j),  1 \leq  j \leq  q$).   This observation  also
  enables parallelized implementations.
\item Lastly,  there is  no loss  of generality  in assuming  a metric
  $d_{uv}$ for  funnel-tree costs;  even if  the costs  were arbitrary
  their metric  completion is  sufficient for determining  the optimal
  funnel-tree.
\end{itemize}

We refer collectively  to the (ball) costs  $c_{uv}$ and (funnel-tree)
metric distances $d_{uv}$ as \emph{weights}. In this paper we consider
two cases -  one, the general case where the  weights can be arbitrary
and two, the special case where  the nodes are embedded in a Euclidean
space and all weights are induced from the embedding.

\subsection{Our Contributions}

\begin{figure}[h]
\vspace{-20pt}
\begin{center}
\begin{tabular}{ | c | c | c | c |}
\hline
   & g2g, any metric & g2s, any metric & g2g, $\ell_2^2$ norm \\ \hline
	Upper & $O(\log^4 n)$ & $2\ln n$ & $O(\log n)$ \\  \hline
	Lower & $\Omega(\log^{2-\epsilon} n)$ & $\Omega(\log n)$  & $(1-o(1))\ln n$ \\  \hline
\end{tabular}
\end{center}
\caption{A summary of upper and lower bounds achieved in the different
  problems.  The lower bound holds for every fixed $\epsilon >0$}
\label{results}
\vspace{-20pt} 
\end{figure}

Our  main results  on the  minimum energy  g2g-anycast problem  are as
follows:
\begin{enumerate}
\item   We  present   a  polynomial-time   $O(\log^4n)$  approximation
  algorithm  for the  g2g-anycast problem  on $n$  nodes with  general
  weights.  We  complement this with an  $\Omega(\log^{2-\epsilon} n)$
  hardness of  approximation, for  any $\epsilon >  0$ (Section
  \ref{sec:general-g2g}).
\item  One   scenario  with  practical  application   is  where  every
  destination group is a singleton set while source groups continue to
  have  more  than  one  node;  we  refer  to  this  special  case  of
  g2g-anycast  as  {\em  g2s  anycast}.   We  present  a  {\em  tight}
  logarithmic   approximation   result    for   g2s-anycast   (Section
  \ref{sec:singleton}).
\item For the realistic scenario where the nodes are embedded in a 2-D
  Euclidean plane with  path-loss exponent $\kappa \geq  2$, we design
  an efficient $O(\log n)$-approximation algorithm Cover-and-Grow, and
  also  establish a  matching  logarithmic  hardness of  approximation
  result (Section \ref{sec:l2sq-g2g}).
\item Lastly, we compare  Cover-and-Grow with 4 alternative heuristics
  on random 2-D  Euclidean
  instances;  we discover  that  Cover-and-Grow does  well  in a  wide
  variety of  practical situations in  terms of both running  time and
  quality,  besides   possessing  provable  guarantees.    This  makes
  Cover-and-Grow  a go-to  solution  for  designing near-optimal  data
  dissemination networks in the wireless infrastructure space (Section
  \ref{sec:simulations}).
\end{enumerate}

\subsection{Related Work}

A variety of power attenuation  models for wireless networks have been
studied  in  the  literature  \cite{Rappaport01}.   Though  admittedly
coarse, the model based on the  path loss exponent (varying from 2, in
free  space to  4,  in  lossy environments)  is  the  standard way  of
characterizing  attenuation  \cite{pathloss-wiki}.   The  problems  of
energy  efficient multicast  and  broadcast in  this  model have  been
extensively                                                    studied
\cite{WieselthierNE00,WieselthierNE01,WanCL+02,LiLHJ07}.   Two  points
worth mentioning in this context are: one, we consider the funnel-tree
as consisting of point-to-point  directional transmissions rather than
an omni-directional broadcast since the nonlinear cost of energy makes
it more economical to relay through  an intermediate node, and two, we
consider only energy spent in transmission but not in reception.

Network   design  problems   are  notoriously   NP-hard.   Over   time
sophisticated  approximation techniques  have been  developed, ranging
from linear  programming and randomized rounding  to metric embeddings
\cite{WilliamsonS11}.  The g2g-anycast problem with general weights is
a  substantial  generalization  including  problems  such  as  minimum
spanning trees,  multicast trees,  broadcast trees, Steiner  trees and
Steiner forests. Even  the set cover problem can be  seen as a special
case  where the  destination groups  are singletons.   The g2g-anycast
also  generalizes   the  much   harder  group  Steiner   tree  problem
\cite{GargKR00,HalperinKKSW07}.


\junk{
\subsection{Outline} We first consider the most general version of the
g2g-anycast    problem    with    arbitrary   weights    in    Section
\ref{sec:general-g2g}.  Section \ref{sec:singleton} treats the special
case where  destination groups are singletons but  weights continue to
be arbitrary.   In Section  \ref{sec:l2sq-g2g} we study  the situation
where  the   nodes  are  embedded  in  Euclidean   space  and  present
Cover-and-Grow   the   central   result   of  this   paper.    Section
\ref{sec:simulations} presents the findings from extensive simulations
comparing Cover-and-Grow  with four alternate  heuristics.  Finally we
conclude  with a  summary of  our findings  and directions  for future
research in Section \ref{sec:conclusion}.
} 

%% file: general-g2g.tex
\section{Approximating g2g-anycast}
\label{sec:general-g2g}

In this  section, we  present an  $O(\log^4 n)$-approximation  for the
g2g-anycast  problem  with  general  weights by  a  reduction  to  the
generalized set-connectivity  problem.  We then give  a reduction from
the  group Steiner  tree problem  that demonstrates  that there  is no
polynomial-time  $\log^{2-\epsilon}   n$-approximation  algorithm  for
g2g-anycast unless $P=NP$.

\subsection{Approximation algorithm for g2g-anycast with general weights}
The generalized set-connectivity  problem~\cite{ChekuriEGS11} takes as
input an edge-weighted  undirected graph $G = (V,  E)$, and collection
of demands $\{(S_1,T_1), \ldots,  (S_k,T_k)\}$, each pair are disjoint
vertex  sets.  The  goal is  to  find a  minimum-weight subgraph  that
contains a path from any node in  $S_i$ to any node in $T_i$ for every
$i \in \{1, \ldots, k\}$. Without loss of generality, the edge weights
can be  assumed to  form a  metric. Chekuri  et al~\cite{ChekuriEGS11}
present an $O(\log^2 n \log^2 k)$-approximation for this problem using
minimum density junction trees.

We show a reduction from  the g2g-anycast problem with general weights
to the generalized set-connectivity problem.  Recall that without loss
of generality, we may  assume that in the g2g problem,  we are given a
single  source group  $S$, a  collection of  destination groups  $T_1,
\ldots, T_q$, nonegative (broadcast) costs $c_{uv}$, and (funnel-tree)
metric costs $d_{uv}$.



\subsection{The Reduction}

\begin{wrapfigure}{r}{0.5\textwidth}
\vspace{-25pt}
\includegraphics[width=0.45\textwidth,trim={0 3cm 0 3.5cm},clip]{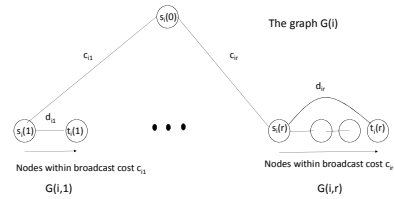}
\caption{A  connected component  $G(i)$  in  the
  reduction  of the g2g-anycast  problem with  general weights  to the
  generalized set connectivity problem. }
\label{fig:relay}
\vspace{-20pt}
\end{wrapfigure}
The main  idea of the reduction  is to overload the  broadcast cost of
the ball radius  around each node in the source group  $S$ into a larger
single  metric  in  which  we  use  the  generalized  set-connectivity
algorithm.  In particular, for every  source node $s_i \in S$, we sort
the  nodes  in $T_1  \cup  \ldots \cup  T_q$  in  increasing order  of
broadcast cost from $s_i$ to  get the sorted order, say $t^i_1, \ldots,
t^i_r$ where $t^i_j$  is at distance $c_{ij}$ from  $s_i$, and we have
$c_{i1} \leq c_{i2} \ldots \leq  c_{ir}$, where $|T_1 \cup
\ldots  \cup T_q| =  r$. We  now build  $r$ different  graphs $G(i,1),
\ldots, G(i,r)$ where  $G(i,j)$ is a copy of  the metric completion of
$G$ under  the funnel tree costs  $d$ induced on the  node set $\{s_i,
t^i_1, \ldots, t^i_j  \}$, with the  copies denoted as $\{s_i(j),
t^i_1(j), \ldots, t^i_j(j) \}$.   (Note that the terminal node $t^i_a$
appears in copies  $a$ through $r$.)  Finally, we  take the $r$ copies
of the node $s_i$ denoted $s_i(1), s_i(2), \ldots, s_i(r)$ and connect
them to a  new node $s_i(0)$ where the cost of  the edge from $s_i(j)$
to  $s_i(0)$ is  $c_{ij}$. Thus  these $r$  different  copies $G(i,1),
\ldots, G(i,r)$ all  connected to the new node  $s_i(0)$ together form 
one connected component $G(i)$.  We  now repeat this process for every
source node  $s_i$ for  $i \in \{1,  \ldots,k\}$ to get  $k$ different
graphs $G(1), \ldots, G(k)$.

We are now  ready to define the generalized  set connectivity demands.
We  define a  new super  source set  $SS =  \{s_1(0), s_2(0),  \ldots,
s_k(0) \}$.  For each of the  destination groups $T_x$, we  define the
terminal set $TT_x$ to be the union of the copies of all corresponding
terminal nodes  in any of  the copies  $G(i)$. More precisely  $TT_x =
\{\cup_i t^i_a(j)| a \leq j \leq r, t^i_a \in T_x\}$. The final demand
pairs  for the  set connectivity  problem are  $\{ (SS,TT_1),  \ldots,
(SS,TT_q)\}$.

\begin{lemma}
  Given an  optimal solution  to the g2g-anycast problem, there  is a
  solution to  the resulting set connectivity  problem described above
  of the same cost.
\end{lemma}
\begin{proof}
Suppose  the solution of  the g2g  problem involved  picking broadcast
ball radii  $c_1, \ldots,  c_k$ from source  nodes $s_1,  \ldots, s_k$
respectively.   We also  have  funnel trees  $H_1,  \ldots, H_k$  that
connect  terminals  $T(H_1), \ldots,  T(H_k)$  to  $s_1, \ldots,  s_k$
respectively.   Note that  all terminals  in $T(H_x)$  are  within the
thresholds that receive the broadcast  from $s_x$, i.e. for every such
terminal $t \in H_x$, the broadcast cost of the edge between $s_x$ and
$t$  is  at  most  the  radius  threshold  $c_x$  at  which  $s_x$  is
broadcasting.

Consider the tree $H_x$ with  terminals $T(H_x)$ connected to the root
$s_x$, so that  $c_x$ is the largest  weight of any of  the edges from
$s_x$ to any  terminal in $T(H_x)$. (If all of  them were even closer,
we can reduce the broadcast cost  $c_x$ of broadcasting from $s_x$ and
reduce the cost  of the g2g solution.) Let the  terminal in the funnel
tree with  this broadcast cost  be $t(x)$ and  in the sorted  order of
weights from $s_x$ let the rank of $t(x)$ be $p$.  We now consider the
graph copy $G(x,p)$ and  take a copy of the funnel  tree $H_x$ in this
copy.  To  this we  add an  edge from  the root  $s_x(p)$ to  the node
$s_x(0)$ of cost  $c_{xp}$. The total cost of this  tree thus contains
the funnel  tree cost of  $H_x$ (denoted by  $d(H_x)$) as well  as the
broadcast  cost of  $c_{xp}$ from  $s_x$.   Taking the  union of  such
funnel trees over all the copies gives the lemma.
\end{proof}

\begin{lemma}
  Given an optimal solution  to the set connectivity problem described
  above, there is a solution to  the g2g-anycast problem from which it
  was derived of the same total weight.
\end{lemma}
\begin{proof}
In the other direction, consider each copy $G(x)$ in turn and consider
the set  of edges in the  tree containing the source  node $s_x(0)$ in
the  solution  to the  generalized  set-connectivity instance.   First
notice that  it contains at most one  of the edges to  a copy $s_x(q)$
for some $q$. Indeed if we have edges to two different copies $s_x(p)$
and $s_x(q)$  from $s_x(0)$  for $p <  q$, then since  $G(x,p) \subset
G(x,q)$, we  can consider the tree  edges in $G(x,p)$ and  buy them in
$G(x,q)$ where they  also occur to cover the same  set of terminals at
smaller cost. In this way, we  can save the broadcast cost of the copy
of the edge from $s_x(0)$  to $s_x(p)$ contradicting the optimality of
the solution. Now that we have  only one of the edges, say to $s_x(q)$
from $s_x(0)$, we  can consider all the edges of the  tree in the copy
$G(x,q)$ and include these edges in a funnel tree $H'_x$. The distance
of the edge  from $s_x(0)$ to $s_x(q)$ pays  for the broadcasting cost
from $s_x$  in the original instance and  the cost of the  rest of the
tree is  the same  as the funnel  tree cost  of $H'_q$ (Note  that our
observation above implies  that edges in the metric  completion in the
tree can be converted to paths  in the graph and hence connect all the
nodes in the tree).

Since
every terminal  superset $TT_j$  is connected to  some source  node of
$SS$, all  the demands  of the  g2g problem must  be satisfied  in the
collection of  funnel trees  $H'_x$ constructed in  this way  giving a
solution to the g2g problem of the same cost.
\end{proof}

The above  two lemmas with the result  of~\cite{ChekuriEGS11} gives us
the following result.
\begin{theorem}
  The  general weights  version of  the g2g-anycast  problem with  $k$
  destination groups admits  a polynomial-time approximation algorithm
  with performance ratio $O(\log^2(k) \log^2n)$ in an $n$-node graph.
\end{theorem}

\subsection{Hardness of approximating g2g-anycast}

We  observe that  the  g2g-anycast problem  with  general weights  can
capture the group Steiner tree problem which is known to be $\log^{2 -
  \epsilon}  n$-hard  to  approximate  unless  $NP$  is  contained  in
quasi-polynomial time~\cite{HalperinK03}.

In the  group Steiner tree problem,  we are given an  undirected graph
with metric  edge costs,  a root  $s$ and  a set  of subsets  of nodes
called  groups, say  $T_1, \ldots,  T_g$, and  the goal  is to  find a
minimum cost tree  that connects the root with at  least one node from
each group. We can easily define  this as a g2g-anycast problem with a
singleton  source group  $S=\{s\}$  with the  single  root node.   The
terminal sets for the g2g-anycast problem are the groups $T_1, \ldots,
T_g$, with the demand graph $(S,T_1), \ldots, (S,T_g)$. We can set the
broadcast costs of any  node in the graph from $s$ to  be zero; we use
the given metric costs in the group Steiner problem as the funnel tree
costs to capture  the cost of the group Steiner  tree. Any solution to
the resulting g2g-anycast  problem is a single tree  connecting $s$ to
at least  one node  in each of  the groups as  required and  its total
weight is just  its funnel tree cost that reflects  precisely the cost
of this  feasible group  Steiner tree  solution. The  hardness follows
from this approximation-preserving reduction.

%% file: singleton-g2g.tex
\section{Approximating g2s-anycast}
\label{sec:singleton}
In  this section,  we  consider  g2s-anycast, a  special  case of  the
g2g-anycast, in which each destination group is a singleton set (i.e.,
has exactly  one terminal).  Let  $S$ denote the source-set  and $t_1,
\ldots, t_q$ denote the terminals.

The desired  solution is  a collection of  broadcast balls  and funnel
trees  $T_v$, each  rooted at  a source  node $v$,  so that  for every
demand $(S, t_j)$, there exists at least one node $v$ in $S$ such that
$t_j \in T_v$.

We  now   present  a  $\Theta(\log  n)$-approximation   algorithm  for
g2s-anycast   problem.    Our   algorithm  iteratively   computes   an
approximation to a minimum density  assignment, which assigns a subset
of as  yet unassigned terminals  to a  source node, and  then combines
these assignments to form the final solution.

{\noindent {\bf  Minimum density assignment.}}   We seek a  source $s$
and  a tree  $T_s$ rooted  at $s$  that connects  $s$ to  a subset  of
terminals, such that the ratio  $(c(T_s) + d(T_s))/|T_s|$ is minimized
among all choices of $s$ and  $T_s$ (here $c(T_s)$ denotes the minimum
broadcast cost for $s$ to reach  the terminals in $T_s$, while $d(T_s$
denotes the  funnel-tree cost,  i.e. the sum  of the  metric distances
$d_{uv}$   over   all   edges   $uv   \in   T_s$).    We   present   a
constant-approximation  to  the  problem,   using  a  constant  factor
approximation  algorithm  for the  rooted  $k$-MST  problem, which  is
defined as follows: given a graph $G$ with weights on edges and a root
node,  determine a  tree of  minimum weight  that spans  at least  $k$
vertices.  The best known approximation factor for the $k$-MST problem
~\cite{RSM+94} is  2~\cite{Garg05}.  We now present  our algorithm for
minimum density assignment.
\begin{itemize}
\item For  each source $s \in  S$, integer $k \in  [1,n]$, and integer
  $r$ drawn from the set $\{c_{st_j|1\leq j\leq q}\}$:
\begin{itemize}
\item Let  $G'$ denote the graph  with vertex set $\{s\}  \cup \{t_j |
  c_{st_j} \le r\}$, and edge weights given by $d$.
\item Compute  a 2-approximation  $T'(s,r,k)$ to the  $k$-MST problem
  over the graph $G'$ with $s$ being the root.
\end{itemize}
\item Among all trees computed in  the above iterations, return a tree
  that minimizes $\min_{s,r,k} (d(T'(s,r,k)) + r)/k$.
\end{itemize}

\begin{lemma}
\label{lem:min_density}
The above algorithm is a polynomial-time 2-approximation algorithm for
the minimum density assignment problem.
\end{lemma}
\begin{proof}
  We first show that the algorithm  is polynomial time.  The number of
  different choices for the source equals  the size of the source set,
  the number  of choices for $k$  is $n$, and the  number of different
  values for $r$ is the number  of different broadcast costs, which is
  at most $n$.   Thus the number of  iterations in the for  loop is at
  most $n^3$.  Consider an optimal solution $T$ to the minimum density
  assignment problem, rooted at source $s$.  It is a valid solution to
  the $k$-MST problem in the iteration given  by $s$, $r = c(T)$, $k =
  |T|$.  For this particular iteration, the tree $T'(s,r,k)$ satisfies
  $(d(T'(s,r,k)+r)/k \le (2d(T) + r)/k \leq 2 \cdot (d(T) + r)/k)$.
  Since our algorithm  returns the tree that has the  best density, we
  have a 2-approximation for the minimum density assignment.
\end{proof}

\noindent{{\bf  Approximation algorithm  for  g2s-anycast.}
  Our  algorithm is  a greedy  iterative
  algorithm, in  which we repeatedly  compute an approximation  to the
  minimum density assignment problem,  and return an appropriate union
  of all of the trees computed.
\begin{itemize}
\item
For each source $s$, set $T_s$ to $\{s\}$.
\item
While all terminals are not assigned:
\begin{itemize}
\item Compute a 2-approximation  $T$ to the minimum density assignment
  problem using any source $s$ and the unassigned terminals.
\item If $T$ is rooted at source $s$, then set $T_s$ to be the minimum
  spanning tree of the union of the trees $T$ and $T_s$.
\end{itemize}
\item
Return the collection $\{T_s\}$.
\end{itemize}

\begin{theorem}
  The  greedy   algorithm  yields  an   approximation  algorithm  with
  performance  ratio  $2\ln   n$  to  the g2s-anycast
  problem.
\end{theorem}
\begin{proof}
  Let $OPT$  denote the cost of  the optimal solution  to the problem.
  Any solution is  composed of at most $m$ trees, one  for each of the
  sources, with each  singleton group being included as  a node in one
  of these trees.  Let $T^*_s$ denote the tree rooted at source $s$ in
  an optimal solution.

Consider any  iteration $i$  of our algorithm.   Let $n_i$  denote the
number of unassigned terminals at  the start of the iteration $i$.  By
an averaging argument, we know there exists a source $s$ such that
\[
\frac{d(T^*_s) + c(T^*_s)}{|T^*_s|} \le \frac{OPT}{n_i},
\]
By Lemma~\ref{lem:min_density}, it follows that in the $i$th iteration
of the  greedy algorithm, if $T_i$  is the tree computed  in the step,
then
\[
\frac{d(T_i) + c(T_i)}{|T_i|} \le \frac{2\cdot OPT}{n_i},
\]
Adding over all steps, we obtain that the total cost is
\[
\sum_i (d(T_i) + c(T_i)) \le 2\cdot OPT \cdot \sum_i \frac{|T_i|}{n_i} \le 2\cdot OPT \cdot H_n \le 2OPT \ln n.
\]
\end{proof}

{\noindent {\bf Hardness of approximation}} We complement the positive
result with  a matching inapproximability result which  shows that the
above problem is as hard as set cover.
\begin{theorem}
  Unless  $NP  =  P$  there  is  no  polynomial-time  $\alpha  \ln  n$
  approximation to  the g2s-anycast  problem, for a  suitable constant
  $\alpha > 0$.
\end{theorem}

We defer the proof of this theorem to Appendix~\ref{app:reduction}. 

%% file: l2squared-g2g.tex
\section{Euclidean g2g-anycast}
\label{sec:l2sq-g2g}
In this section,  we present a $\Theta(\log  n)$-approximation for the
more realistic version of the g2g-anycast problem in the 2-D Euclidean
plane.   We achieve  our results  by a  reduction to  an appropriately
defined set cover problem.

In detail, all the points in both the source group $S$ and destination
groups $T_1, \ldots, T_q$ lie in  the 2-D Euclidean plane. The cost of
an edge $(u,v)$  is the Euclidean distance between $u$  and $v$ raised
to the path  loss exponent $\kappa$. For the rest  of this section, we
assume that $\kappa = 2$. (The  corresponding results for $\kappa > 2$
follow with very  simple modifications.) First we show  that even this
special  case   of  the  g2g-anycast   problem  does  not   permit  an
approximation algorithm with ratio  $(1-\epsilon)\ln n$ on an instance
with  $n$ nodes  unless $NP$  is in  quasi-polynomial time.   Next, we
present {\em  Cover-and-Grow}, an $O(\log  n)$-approximation algorithm
that applies a  greedy heuristic to an  appropriately defined instance
of the set covering problem.

{\noindent {\bf  Hardness of  2-D g2g-anycast}} Again  we can  prove a
hardness via a reduction from set cover.
\begin{theorem}
  The 2-D Euclidean  version of the g2g problem on  $n$ nodes does not
  permit  a polynomial-time  $(1-o(1))\ln  n$ approximation  algorithm
  unless $NP = P$.
\end{theorem}

The proof of this is deferred to Appendix~\ref{app:reduction2}.

\subsection{Cover-and-Grow}
We now describe a matching  $O(\log n)$-approximation for the problem.
For  this we  first need  the following  property of  minimum spanning
trees of points in the 2-D  Euclidean plane within a unit square, when
the costs of any edge in  the tree are the squared Euclidean distances
between the edge's endpoints.

\begin{theorem}~\cite{AichholzerAAB+}
\label{thm:euc}
The weight of a minimum spanning tree  of a finite number of points in
the 2-D Euclidean plane within a  unit square, where the weight of any
edge is the square of the Euclidean distance between its endpoints, is
at most 3.42.
\end{theorem}

We can apply this theorem to bound the cost of the funnel trees within
any demand ball in the solution within a factor of at most 3.42 of the
cost of the  ball. Indeed, by scaling the diameter  of the demand ball
to correspond to  unit distance, the above theorem shows  that for any
finite set  of terminal  nodes (i.e. nodes  in the  destination group)
within the  ball, a funnel  tree which is  an MST that  connects these
terminal  nodes to  the center  of  the ball  has total  cost at  most
3.42.  The cost  of the  demand ball  is the  square of  the Euclidean
distance of  the ball radius  which, in  the scaled version,  has cost
$(\frac{1}{2})^2 = \frac{1}{4}$.  This shows that the  funnel tree has
cost at most  13.68 times the cost of the  funnel ball. This motivates
an algorithm that uses balls of  varying radii around each source node
as a ``set" that has cost equal  to the square of the ball radius (the
ball cost) and  covers all the terminal nodes within  this ball (which
can be connected in a funnel tree  of cost at most 13.68 times that of
the demand ball).

{\em Algorithm Cover-and-Grow}
\begin{enumerate}
\item Initialize the solution to be empty.
\item While there is still an unsatisfied demand edge
\begin{itemize}
\item For every source node $s_i$,  for every possible radius at which
  there is a terminal node belonging to some destination group $T$ for
  which the  demand $(S,T)$ is  yet unsatisfied, compute the  ratio of
  the square of the Euclidean radius of the ball to the number of {\em
    as yet unsatisfied} destination groups whose terminal nodes lie in
  the ball.
\item Pick  the source  node and  ball radius  whose ratio  is minimum
  among all the available balls, and  add it to the solution (both the
  demand ball around this node and a funnel tree from one node of each
  destination group whose demand is unsatisfied at this point). Update
  the set of unsatisfied demands accordingly.
\end{itemize}
\end{enumerate}

\begin{theorem}
  Algorithm  Cover-and-grow  runs  in  polynomial  time  and  gives  an
  $O(\log n)$-approximate solution  for the  2-D
  g2g-anycast problem in an $n$-node graph.
\end{theorem}

\begin{proof}
  We will use a reduction from the given 2-D g2g-anycast problem to an
  appropriate set  cover problem  as described  in the  algorithm: The
  elements of the  set cover problem are the terminal  sets $T_j$ such
  that the demand graph has the  edge $(S,T_j)$. For every source node
  $s_i\in S$,  and for every possible  radius $r$ at which  there is a
  terminal  node belonging  to some  destination group  $T$ for  which
  there  is  a demand  $(S,T)$,  we  consider  a set  $X(s_i,r)$  that
  contains all  the destination  groups $T_j$ such  that some  node of
  $T_j$ lies within this ball. The cost of this set is $r^2$.

  First, we  argue that  an optimal solution  for the  2-D g2g-anycast
  problem of cost $C^*$ gives a solution of cost at most $C^*$ to this
  set cover problem.   Next, we show how any feasible  solution to the
  set cover problem  of cost $C$ gives a feasible  solution to the 2-D
  g2g-anycast problem of cost at most $14.68C$. These two observations
  give us the  result since the algorithm we describe  is the standard
  greedy approximation algorithm for set cover.

  To see the first observation, given  an optimal solution for the 2-D
  g2g-anycast problem of cost $C^*$, we pick the sets corresponding to
  the demand  balls in the solution  for the set cover  problem. Since
  these demand balls  are a feasible solution to  the anycast problem,
  they  together  contain at  least  one  terminal  from each  of  the
  destination  groups  $T_j$   for  which  there  is   a  demand  edge
  $(S,T_j)$. These balls form a solution  to the set cover problem and
  the demand  ball costs  of the  anycast solution  alone pay  for the
  corresponding costs of  the set cover problem.   Hence this feasible
  set cover solution has cost at most $C^*$.

  For  the other  direction, given  any feasible  solution to  the set
  cover problem of cost $C$, note  that this pays for the demand balls
  around the source nodes in this  set cover solution.  Now we can use
  the implication in the  paragraph following Theorem~\ref{thm:euc} to
  construct a funnel tree for each of these demand balls that connects
  all  the terminals  within these  balls to  the source  node at  the
  center of  the ball with  cost at most 13.68  times the cost  of the
  demand ball around  the source node. Summing over all  such balls in
  the solution gives the result.
\end{proof} 

%% file: shortsimulations.tex
\section{Empirical Results}
\label{sec:simulations}

We conducted simulations comparing  Cover-and-Grow with four different
natural heuristics  for points embedded  in a  unit square in  the 2-D
Euclidean plane.  These  simulations allow us gain  perspective on the
real-world utility  of Cover-and-Grow vis  a vis alternatives  that do
not  possess provable  guarantees but  yet  have the  potential to  be
practical.  The  specifics of  the simulation and  the details  of the
results                are                 discussed                in
Appendix~\ref{app:simulations}. Cover-and-Grow  performs comparably to
the heuristics in  performance; and the runtime  of Cover-and-Grow was
better than the heuristics except for the T-centric approach.

%% file: setCoverreduction.tex
\section{g2s-anycast Hardness}
\label{app:reduction}

\begin{theorem}
  Unless  $NP  =  P$  there  is  no  polynomial-time  $\alpha  \ln  n$
  approximation to  the g2s-anycast  problem, for a  suitable constant
  $\alpha > 0$.
\end{theorem}
\begin{proof}
  Our proof is by a reduction from the minimum set cover problem.  Let
  ${\cal U}$  denote a collection  of elements, and let  $X_1, \ldots,
  X_m$ denote the  sets containing elements from ${\cal  U}$.  The set
  covering objective is to minimize  the number of subsets whose union
  is  ${\cal  U}$.   We  construct   the  following  instance  of  the
  g2s-anycast  problem.   The  terminal   nodes  (i.e.  nodes  of  the
  destination  groups) correspond  to the  elements of  the set  cover
  instance.   The source  nodes correspond  to the  sets.  We  set the
  funnel-tree distance  between a source  $X_i$ and an element  $e$ in
  $X_i$ to be $1$.  The remaining distances are captured by the metric
  completion of  these distances.   We next  set the  broadcast costs.
  For any source node $X_i$ and element  $e$ in $X_i$, we set the cost
  $c(X_i, e)$ to  be $L$, for a  suitably large $L \gg n$;  for all $e
  \notin X_i$, we set $c(X_i, e)$ to $M \gg L$.

  If there is a solution of cost  $C$ to the set cover instance, there
  is a solution of  cost $CL + n$ to the  g2s-anycast problem.  On the
  other hand,  consider any solution  to the g2s-anycast  problem.  It
  incurs a  funnel tree cost  of at least  $n$.  If, in  the solution,
  every source node broadcasts only to terminal nodes corresponding to
  its set  in the set cover  instance, then the broadcast  cost equals
  $L$ times  the cost  of the  resulting set  cover obtained  by those
  source nodes  that broadcast to at  least one terminal; we  refer to
  such as a  solution as a canonical solution.  If  the solution has a
  source node that broadcasts to  a terminal outside its corresponding
  set, then the broadcast cost is at least $M$.  Given that there is a
  solution to the set cover instance using all of the sets, there is a
  canonical solution to the anycast problem  of cost at most $mL + n$.
  By selecting $M$ to  be $\Omega(m L \ln n)$, we  can ensure that any
  $O(\lg n)$-approximation to the  g2s-anycast instance will produce a
  canonical solution.  From  this, we obtain that if we  set $L$ to be
  sufficiently larger  than $n$, any  $\alpha \ln n$  approximation to
  the  g2s-anycast  instance  yields  an $(\alpha  -  \epsilon)\ln  n$
  approximation, for  an $\epsilon >  0$ that can be  made arbitrarily
  small  by  making $L$  sufficiently  large.   This, along  with  the
  hardness for set cover in~\cite{DinurS14}, completes the proof of the
  theorem.
\end{proof}

\section{Euclidean g2g Hardness}
\label{app:reduction2}
Recall that in the  set covering  problem, we are  given a  ground set $E$  of $n$
elements and a collection of subsets $X_1, \ldots, X_m \subseteq E$ of
elements, and  the goal is to  find a minimum number  of these subsets
whose union is $E$.
We  present  an
approximation-preserving reduction  from a  given instance of  the set
covering problem to one of the 2-D g2g-anycast as follows.

For each subset $X_i$ we pick a point $x_i$ in the plane such that any
pair of such ``set-points" are quite far from each other (distance $>>
n$)  in the  plane.  Our source  set  $S$ will  consist  of these  $m$
different set-points.

For each  point $x_i$,  we pick  a point $y_i$  at unit  distance from
$x_i$ in  the plane to place  copies of the element-points.   For each
element $e \in E$, we create a destination group $T_e$, which consists
of as many nodes as the number of  sets in which $e$ occurs. If $e \in
X_i$, then  we create a terminal  node $t(e)_i$ at the  point $y_i$ in
the  plane. Note  that all  elements that  belong to  a set  $X_i$ are
co-located  in the  point $y_i$  at unit  distance from  the set-point
$x_i$. The demand  graph for the resulting g2g-anycast  problem is all
pairs of  the form $(S,T_e)$  for every element  $e$ in the  set cover
problem.  The following observations are now immediate.

\begin{lemma}
  Given an optimal solution to  the set cover problem with $k^*$ sets,
  there is a  solution to the 2-D g2g-anycast  problem using the
  above reduction of cost $2k^*$.
\end{lemma}
\begin{proof}
To convert an optimal solution of  the set cover problem, for each set
$X_i$ in the optimal set cover, we pick the set-point $x_i$ and draw a
unit ball around it, which encloses the point $y_i$ containing all the
terminal  nodes  corresponding  to   elements  contained  in  the  set
$X_i$. For  all these element  terminal nodes co-located at  $y_i$, we
build a funnel tree of a single edge from $y_i$ back to $x_i$. The sum
of the Euclidean length squared costs of the ball around $x_i$ and the
funnel tree is  two. Repeating for every set  in the optimal solution,
we get a solution to the g2g-anycast problem.
\end{proof}

\begin{lemma}
  Given an  optimal solution  to the  2-D g2g-anycast  problem arising
  from a reduction from a set cover problem as described above of cost
  $C$, there is a solution to the  set cover from which it was derived
  that contains at most $\frac{C}{2}$ sets.
\end{lemma}
\begin{proof}
Observe  that all  minimal solutions  correspond to  unit-radius balls
around a  set of set-points  $x_i$, and the  funnel trees for  each of
these points in  the solution all consist of a  single edge from $y_i$
to $x_i$.   Since the  g2g-anycast solution  is feasible,  these balls
around the set-points cover all demands  and hence form a feasible set
cover. The number of sets in the solution is exactly $\frac{C}{2}$.
\end{proof}
We now  get the  following lower bound  on the approximability  of the
problem using~\cite{DinurS14}.

\begin{theorem}
  The 2-D Euclidean  version of the g2g problem on  $n$ nodes does not
  permit  a polynomial-time  $(1-o(1))\ln n$  approximation algorithm,
  unless $NP = P$.
\end{theorem}

%% file: simulations-g2g.tex
\section{Empirical Results}
\label{app:simulations}

We conducted simulations  comparing Cover-and-Grow with four different
heuristics for  points embedded in a unit square. Both broadcast costs 
and the funnel-tree costs are assumed to be the square of the Euclidean distances (using the path-loss exponent value of $\kappa = 2$). 

In  our  simulations we  had  one $S$  group  (this  is sufficient  as
mentioned in  Section \ref{subsec:formulation}) and  for $|S|$ we
chose  $1,4,16$ and  $64$.   We had  $10$  $T$ groups  each with  $10$
terminals.   We   ran  our  trials   on  two  different   basic  point
distributions; the  uniform distribution over a unit square in the plane and  a  Gaussian distribution in the whole two dimensional plane. Our  results are  averaged over  $100$
trials for each choice of  parameter settings. The variance across the
trials was  negligible (and so we do  not show any error  bars as they
would only clutter our plots).  The  simulations were run on
an enterprise  class server with  an Intel(R) Core(TM)  i7-4500U (dual
core) CPU  @ 3.0 GHz  Turbo with  32 GB of  RAM.  The entire  suite of
simulations  took   50  hours  to  complete.   Due   to  the  inherent
combinatorial explosiveness of  g2g-anycast it was entirely infeasible
to compute the  optimal solution; therefore, in our  figures we depict
the quality of the solutions of the heuristics relative to the quality
of Cover-and-Grow normalized to $1$.

We now  describe the  four heuristics we  implemented (in  addition to
Cover-and-Grow).  The descriptions  below only detail the construction
of the  funnel-trees since it follows  that in a  minimal solution the
ball at each $s \in S$ node will be the smallest one that encloses the
funnel-trees containing $s$.
\begin{itemize}
\item \emph{Smallest  Edge} repeatedly adds the smallest edge not yet
  in the set  which does not create a cycle or a component
  with two $S$ nodes. 
  The process stops  when every $T_j$ has a vertex in some
  component with an $S$ node.  It then (repeatedly) removes all edges that
  are in a component with no $S$ node, as well as the largest edges whose
  removal would not result in a disconnection of any $T_j$ from $S$.
  Note that this heuristic requires re-computation of shortest distances between sets (current components) and nodes at every iteration which can make it quite time-consuming.
\item \emph{T-Centric}  for each $T_j$, finds the  pair $s\in S,  t\in T_j$  such that
  $d(s,t)$ is minimized and assigns $t$  to $s$. For each $s$, this process
	builds an MST on $s$ and the nodes $t$ which were assigned to $s$.
\item \emph{T-Adaptive} grows clusters  starting with each $s$ in its
  own cluster.  It repeatedly finds  the closest pair $r,t$  such that
  $r$ is in a cluster and $t$ is in a $T_j$ none of whose nodes are in
  a cluster yet and adds edge $r,t$. If we think of \emph{Smallest  Edge} 
  as a loose analog of Kruskal's algorithm for MSTs tailored to our problem, then \emph{T-Adaptive} would be the corresponding Prim variant. Unlike 
  \emph{Smallest  Edge}, which uses Steiner nodes, this heuristic requires only distances from any source cluster and a unassigned terminal node making it less intensive computationally.
\item \emph{Smallest Increment} grows  clusters starting with each $s$
  in it's own cluster. It repeatedly  finds the $r,t$ such that $r$ is
  in a  cluster and $t$  is in a  $T_j$ none of  whose nodes are  in a
  cluster yet  and such that attaching  $t$ to $r$  using the shortest
  path increases  the total cost of  the solution (i.e. funnel-tree  cost as
  well as ball cost) the least. Just like \emph{Smallest  Edge}, this heuristic also requires re-computation of shortest paths between clusters and nodes at every step, making it potentially time intensive. This heuristic is also similar to \emph{T-Adaptive} in growing from source clusters but the differences are the consideration of not just direct edges but shortest paths, as well as the additional increase due to the broadcast cost.
\end{itemize}

\begin{figure}[t]
\centering
\begin{subfigure}[b]{0.45\textwidth}
\includegraphics[width=\textwidth]{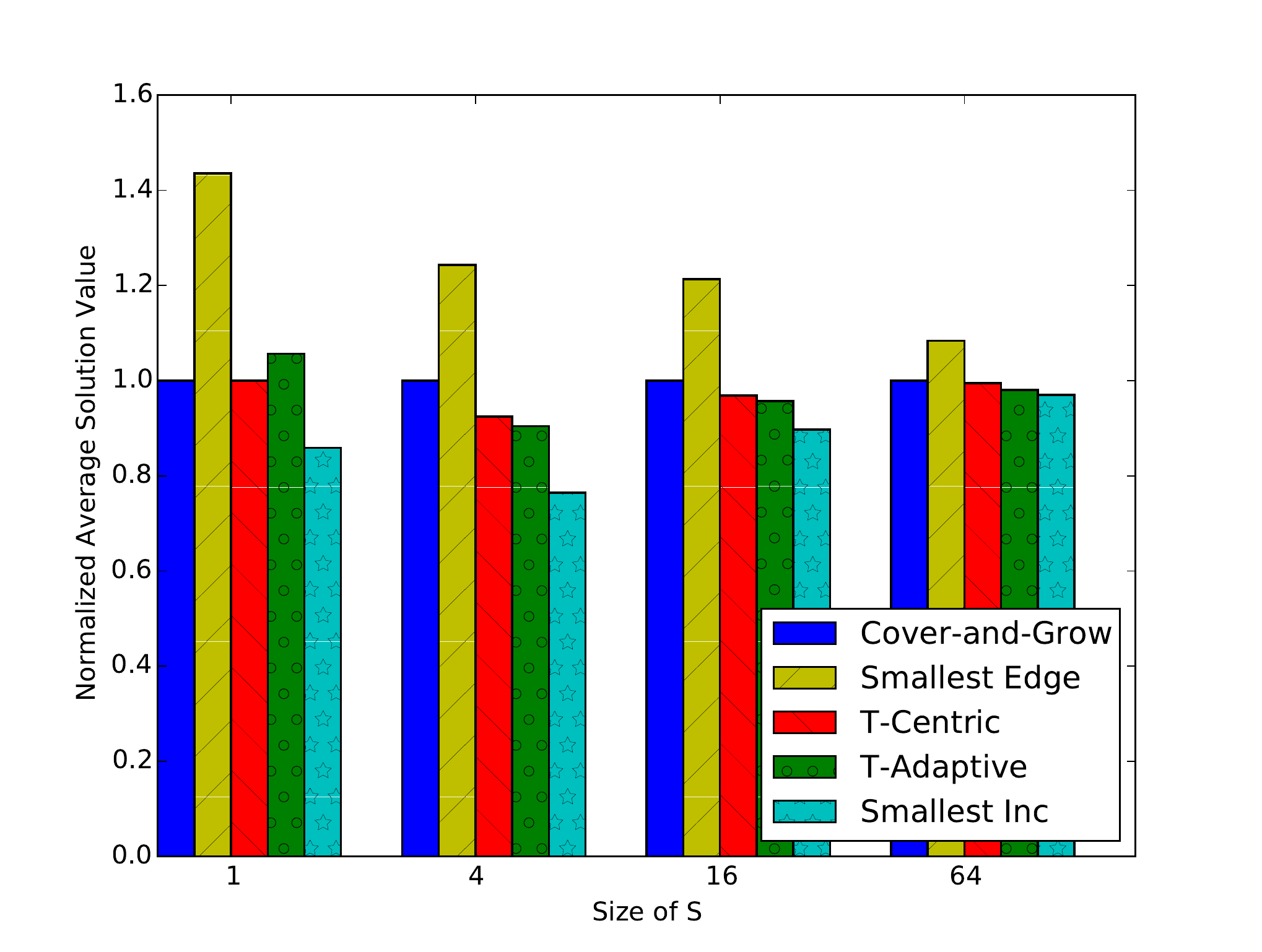}
\caption{Uniform distribution. }
\label{fig:uniNorm}
\end{subfigure}
\begin{subfigure}[b]{0.45\textwidth}
\includegraphics[width=\textwidth]{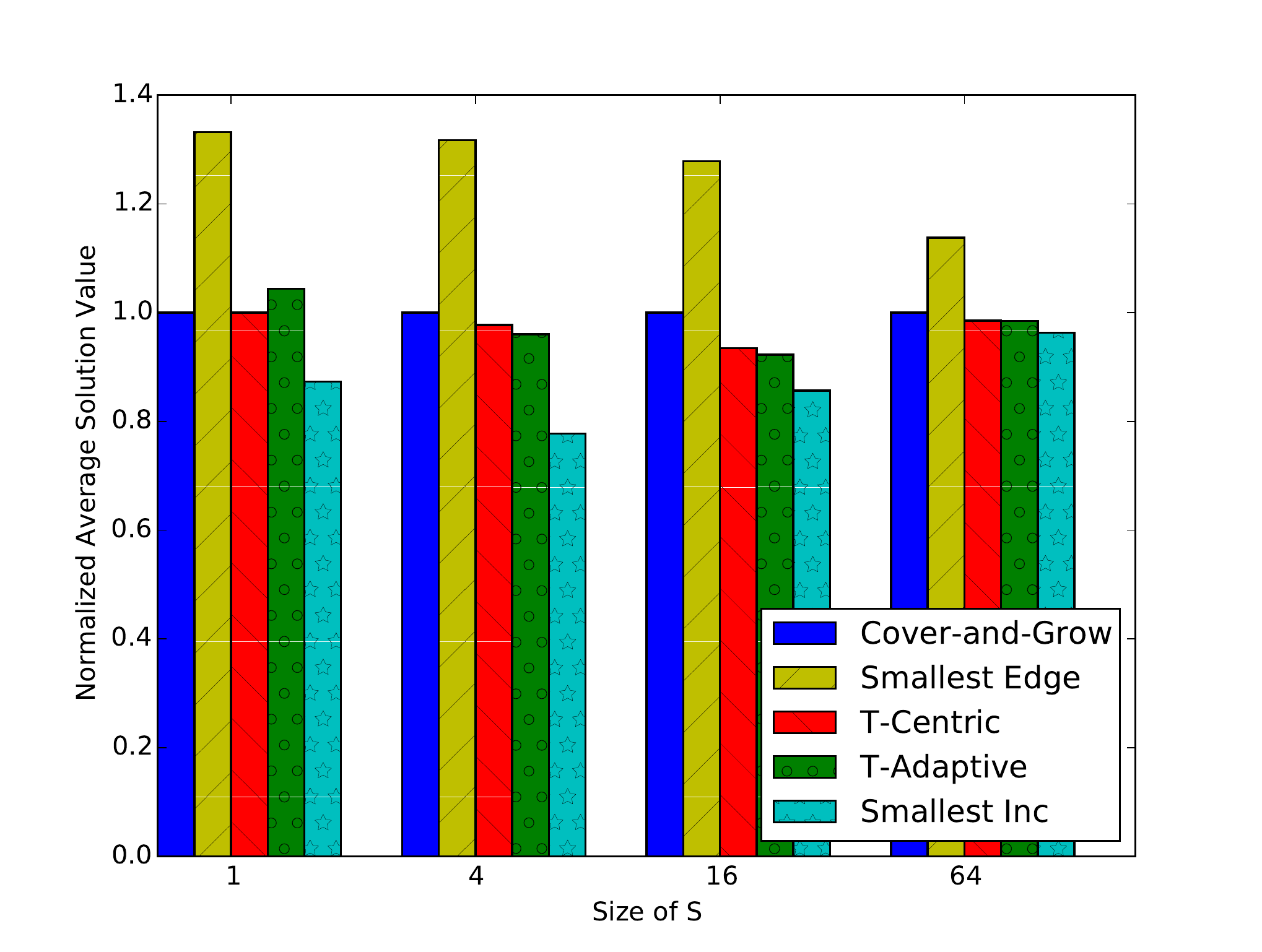}
\caption{Gaussian distribution. }
\label{fig:gausNorm}
\end{subfigure}
\caption{Relative costs of the results of the algorithms on two different distributions. }\label{fig:results}
\end{figure}

Figures~\ref{fig:uniNorm} and  \ref{fig:gausNorm} show the  quality of
the solution of the  heuristics relative to Cover-and-Grow (normalized
to  $1$)  under  the   two  distributions  of  points.  Cover-and-Grow
performed as  well as  the heuristics, losing  out only  marginally to
Smallest  Increment.  Figure~\ref{fig:run1}  shows
the   (absolute)  runtimes   of  the   heuristics  under   the  uniform
distribution.   (We  do  not   show  the  runtimes  for  the  Gaussian
distribution  since the  plot  is  identical.) The runtimes plot was plotted
on a logrithmic scale due to the large differences in the runtime.
{\em Smallest Edge}  and  {\em Smallest  Increment} both
incorporate Steiner nodes that have the potential to greatly
reduce the weight  of the funnel trees with  respect to $\ell_2^2$.
On the flip side, allowing  Steiner nodes  increases the  runtime  of these
algorithms by up to a factor $n$.

\begin{wrapfigure}{r}{0.5\textwidth}
\vspace{-20pt}
\includegraphics[width=0.42\textwidth]{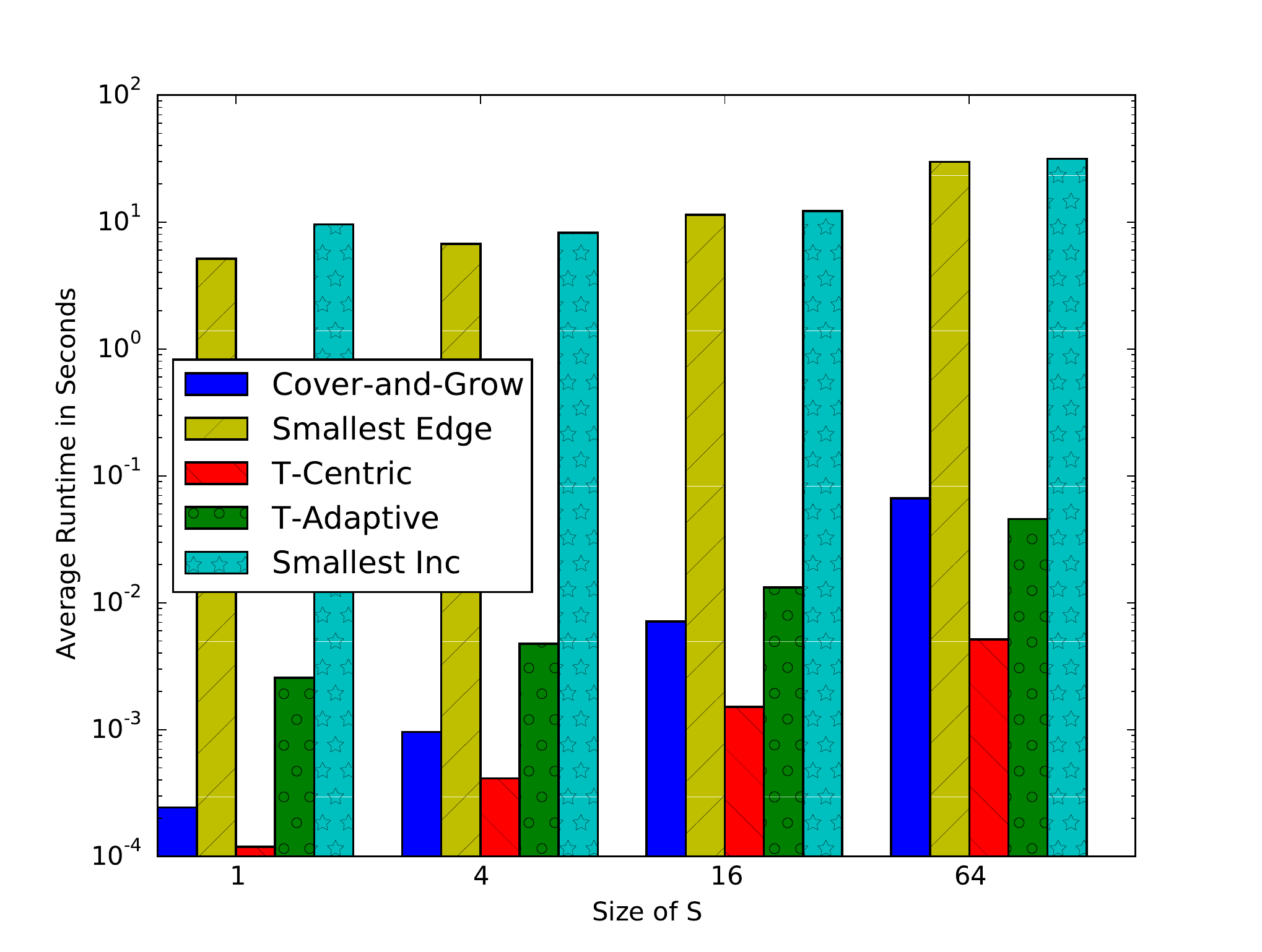}
\caption{Absolute runtimes in seconds for the Uniform Distribution. Runtime for the Gaussian distribution was identical, as the algorithms did not depend on the distribution of points.}
\label{fig:run1}
\vspace{-10pt}
\end{wrapfigure}

The  picture that emerges  from the  plots is  that Smallest  Edge and
Smallest Increment are impractical time-wise. T-Adaptive performed well on the larger instances and scaled well time-wise, but did not do as well as Cover-and-Grow on small instances. This leaves T-Centric  which is  about  the  same as  Cover-and-Grow  in terms  of
quality of solution, and in general runs much faster than Cover-and-Grow. T-centric is by far the fastest as it does not depend on what has been already added or steiner nodes, so it can make all the assignments in one iteration. However,  the following simple example shows that
T-Centric can be as much as  a factor $q$ worse than the optimal:
consider the unit square  with corners $(0,0), (0,1),(1,0),(1,1)$; let
$S = \{(\frac{1}{q},0),(\frac{2}{q},0),\ldots,(1,0)\}$ and let $1 \leq
i  \leq  q,  T_i  =  \{(\frac{i}{q},1),(0,1)\}$; it  is  easy  to  see
T-Centric's  solution (ball and  funnel-trees) is  $\Omega(q)$ whereas
the optimal  is $O(1)$. Thus,  not only does Cover-and-Grow  come with
provable guarantees but in practice, it is superior to the natural
alternative that we have  been able  to come up  with. This  begs the
question -  why does Cover-and-Grow do  so well even though  it too is
myopic?  We believe that the answer  lies in  the fact that  by focusing  on the appropriate density ratio it is greedy in an intelligent way avoiding corner case like the one depicted above.